\documentclass[11pt]{article}

\usepackage{newlfont}
\usepackage[small]{caption2}
\usepackage[below]{placeins}
\usepackage{flafter}
\usepackage{amsfonts}
\usepackage{amsmath}
\usepackage{amssymb}
\usepackage[american]{babel}
\usepackage{graphicx}
\usepackage{subfigure}
\usepackage{multirow}
\usepackage{url}
\usepackage[linesnumbered,ruled,vlined]{algorithm2e}
\usepackage{color}
\usepackage{tikz}
\usepackage{verbatim}

\allowdisplaybreaks

\newcommand{\ket}[1]{\big| #1 \big\rangle}
\newcommand{\bra}[1]{\big\langle #1 \big|}
\newcommand{\braket}[2]{\big\langle #1 \big| #2 \big\rangle}                 

\newtheorem{theorem}{Theorem}[section]

\newtheorem{proposition}[theorem]{Proposition}
\newtheorem{corollary}[theorem]{Corollary}
\newtheorem{definition}[theorem]{Definition}

\newenvironment{proof}[1][Proof]{\begin{trivlist}
\item[\hskip \labelsep {\bfseries #1}]}{\end{trivlist}}

\newcommand{\qed}{\nobreak \ifvmode \relax \else
      \ifdim\lastskip<1.5em \hskip-\lastskip
      \hskip1.5em plus0em minus0.5em \fi \nobreak
      \vrule height0.75em width0.5em depth0.25em\fi}

\begin{document}


\title{The Staggered Quantum Walk Model}
\author{R. Portugal\mbox{$^{1}$}\footnote{Corresponding author: portugal@lncc.br}, R.A.M. Santos\mbox{$^{1}$}, T.D. Fernandes\mbox{$^{1,2}$}, D.N. Gon\c{c}alves\mbox{$^{3}$} \\
\\
{\small $^1$National Laboratory of Scientific Computing - LNCC} \\
{\small Av. Get\'{u}lio Vargas 333, 25651-075, Petr\'{o}polis, RJ, Brazil}\\
{\small }\\
{\small $^2$Universidade Federal do Esp\'{i}rito Santo - UFES, 29500-000, Alegre, Brazil}\\
{\small }\\
{\small $^3$Centro de Educa\c{c}\~ao Tecnol\'{o}gica Celso Suckow da Fonseca - CEFET}\\ {\small 25620-003, Petr\'{o}polis, Brazil}
}

\maketitle

\begin{abstract}
There are at least three models of discrete-time quantum walks (QWs) on graphs currently under active development. In this work we focus on the equivalence of two of them, known as Szegedy's and staggered QWs. We give a formal definition of the staggered model and discuss generalized versions for searching marked vertices. Using this formal definition, we prove that any instance of Szegedy's model is equivalent to an instance of the staggered model. On the other hand, we show that there are instances of the staggered model that cannot be cast into Szegedy's framework. Our analysis also works when there are marked vertices. We show that Szegedy's spatial search algorithms can be converted into search algorithms in staggered QWs. We take advantage of the similarity of those models to define the quantum hitting time in the staggered model and to describe a method to calculate the eigenvalues and eigenvectors of the evolution operator of staggered QWs.
\end{abstract}

\section{Introduction}

Quantum walks (QWs) are being used with many flavors with the goal of understanding quantum systems and building quantum algorithms for quantum computers. Research in QWs started with the coined discrete-time QW~\cite{Aharonov:1993} proposed with the goal of displaying a quantum system with features strikingly different from classical systems. From a different viewpoint, Meyer~\cite{Meyer96} proposed a non-trivial example of a quantum cellular automata which can be somehow considered the starting point of staggered QWs. By quantizing continuous classical Markov chains, Farhi and Gutmann~\cite{Farhi:1998} proposed a continuous-time version of QWs, with clear implications to the area of quantum computing. Inspired by classical discrete Markov chains, Szegedy~\cite{Szegedy:2004} proposed a coinless discrete-time QW model and was able to provide us with a natural definition of quantum hitting time.   

Szegedy's QWs are obtained by quantizing classical discrete Markov chains described by some transition matrix. Szegedy~\cite{Szegedy:2004} also developed QW-based search algorithms inspired by a previous known coined-based search algorithm~\cite{Ambainis:2004}. On ergodic Markov chains, it is possible to detect the presence of a marked vertex at a hitting time that is quadratically smaller than the classical average hitting time~\cite{Magniez:2011}. Szegedy's model was also used for the searching problem~\cite{Magniez:2011,Krovi:2010}, which aims to find the location of a marked vertex, and for searching triangles~\cite{mss07}.  Some references highlighted the power of the method, for instance see Ref.~\cite{Mosca:2009}. The fact that Szegedy's model is a QW on the edges of a bipartite graph rather than on the vertices is emphasized in~\cite{Santha:2008,Magniez:2011}. In this paper we investigate this idea by employing the line graph of the bipartite graph. 

Meyer~\cite{Meyer96} circumvented the triviality of quantum cellular automata by relaxing the homogeneity condition. He proposed an one-dimensional quantum cellular automata driven by staggered unitary operators, which can be transposed as a coinless QW on the line or cycle. The connection of this model with the staggered lattice fermion formalism was analyzed in~\cite{Patel05} for the two-dimensional case and applied for spatial searching in~\cite{Patel:2010}. The staggered model was rediscovered by Falk~\cite{Falk:2013}, who suggested a simple method of obtaining the evolution operator by splitting the vertices of the graph into disjoint polygons that tessellate the two-dimensional lattice. Falk also proposed a searching method, which was used by~\cite{Ambainis:2013} to prove analytically that the staggered model finds a marked vertex in time $O(\sqrt{N\log N})$ for a two-dimensional lattice with $N$ vertices matching the performance achieved by other forms of QWs on this graph. The detailed dynamics of the one-dimensional staggered model was analyzed in~\cite{PBF15} and the moments of the probability distribution function in~\cite{Santos:2015}. The connection between the coined and the staggered one-dimensional QWs was analyzed in~\cite{HKS05,PBF15,Strauch:2006}. Further analysis on this connection is an important research topic.

In this paper we formally define the staggered QW model on graphs by using tessellations so that each tessellation covers all vertices with non-overlapping polygons and the tessellation union covers all edges. The maximal cliques play an essential role in building the tessellation because adjacent vertices are reachable for the walker's hopping. If two or more maximal cliques share a common vertex, each clique must be inside a different tessellation. The evolution operator is a product of orthogonal reflections that have a one-to-one correspondence with the tessellations producing a diffusion through the maximal cliques in a staggered way. After the action of the evolution operator with a localized initial condition, the wave function spreads to all vertices reachable by the graph structure, which is the clique with the starting location and the adjacent maximal cliques for the case with two reflections. We also discuss a generalized form of staggered QW using partial tessellations which can be used for spatial search algorithms in the sense that the walker will search for the vertices inside the missing polygons. The complete dynamics is obtained by applying the evolution operator recursively starting with some initial state and by performing a measurement at the end. 

After formally defining the staggered model, we show that Szegedy's QWs on a bipartite graph $\Gamma$ are staggered QWs on the line graph of $\Gamma$ with a restricted form of tessellation. The restriction is: The polygons must share only one vertex. The converse is also true. Any staggered QW with tessellations that share one vertex in the polygon intersections can be cast into an extended version of Szegedy's model with complex amplitudes. This result sheds light on the structure of Szegedy's model and explains the meaning of the fact that a Szegedy's QW is a hopping on the edges of the bipartite graph rather than on the vertices. We also show that there are generalized staggered QWs that reproduce the searching model proposed by Szegedy. In this case we need to use partial tessellations and the walker looks for vertices inside missing polygons, which are the ones returned by the measurement. On the other hand, if we allow two or more vertices in at least one polygon intersection, we define instances of staggered QWs that cannot be cast into Szegedy's framework. We give special attention to those instances.

Falk~\cite{Falk:2013} proposed an evolution operator for a searching model on the two-dimensional lattice by interlacing a reflection around the marked vertices with the orthogonal reflections generated by the tessellations. We use this searching model to define the hitting time on staggered QWs on finite graphs by generalizing Szegedy's definition and we describe how to obtain the eigenvalues and eigenvectors of the evolution operator using the discriminant matrix generated by the inner product of the polygons of the tessellations. We show that the evolution operator with marked vertices can be written as a product of two orthogonal reflections, and if an entire polygon is marked, the walker is not able to find it.

The structure of the paper is as follows. In Sec.~\ref{sec:Staggered} we define the staggered QW model and discuss generalizations that are useful for searching algorithms. In Sec.~\ref{sec:Szegedy} we review Szegedy's QW model with the goal of finding the connection with the staggered QW model. In Sec.~\ref{sec:Szegedisstag}
we show that Szegedy's model is a restricted form of staggered QWs and we characterize when staggered walks cannot be cast into Szegedy's framework.  In Sec.~\ref{sec:HT} we 
define the hitting time on the staggered model using reflection operators around the marked vertices. In Sec.~\ref{sec:SearchusingStag} we analyze the eigenvalues and eigenvectors of the evolution operator for spatial search algorithms based on reflections around the marked vertices. In Sec.~\ref{sec:Conc} we draw our conclusions.

\

\section{Staggered Quantum Walks}\label{sec:Staggered}

Inspired by Falk's paper~\cite{Falk:2013}, we define the staggered quantum walk on a graph $\Gamma$ with $N$ vertices\footnote{$N$ can be infinite.} by using two (or more) independent graph tessellations. Each tessellation uses non-overlapping (non-planar) polygons of adjacent vertices which need not to have the same shape and a polygon may contain only one vertex. The set of polygons of each tessellation need to cover all vertices of the graph and the tessellation union must cover all edges. Two polygons in different tessellations necessarily overlap and some (or all) intersections may contain more than one vertex. Each polygon defines a unit vector in the Hilbert space ${\cal H}^{N}$ by superposing the vertices in the polygon with nonzero amplitudes. Examples of tessellations are depicted in Figs.~\ref{fig:tess} and~\ref{fig:grid}.


The staggered quantum walk with two tessellations called $\alpha$ and $\beta$ is defined by the evolution operator
\begin{equation}\label{U}
    U \,=\, U_1 U_0,
\end{equation}
where
\begin{eqnarray}
  U_0 &=& 2\sum_{k=0}^{m-1} \ket{\alpha_k}\bra{\alpha_k} - I, \label{U_0}\\
  U_1 &=& 2\sum_{k=0}^{n-1} \ket{\beta_k}\bra{\beta_k} - I, \label{U_1}
\end{eqnarray}
and $m$ and $n$ are the number of polygons in each tessellation, and
\begin{eqnarray}
  \ket{\alpha_k} &=&  \sum_{k'\in \alpha_k} a_{k,k'} \ket{k'}, \label{alpha_k} \\
  \ket{\beta_k} &=&  \sum_{k'\in \beta_k} b_{k,k'} \ket{k'}, \label{beta_k}
\end{eqnarray}
where $a_{k,k'}$ and $b_{k,k'}$ are nonzero complex amplitudes of the unit vectors $\ket{\alpha_k}$ and $\ket{\beta_k}$ in ${\cal H}^{N}$. Operators $U_0$ and $U_1$ are unitary and Hermitian $\big(U_{0,1}^2=I\big)$ because each set of polygons is non-overlapping $\big(\braket{\alpha_k}{\alpha_{k'}}=\braket{\beta_k}{\beta_{k'}}=\delta_{k k'}\big)$.
It is straightforward to show that
\begin{equation}
U\ket{k} = 4 \sum_{k',k''} D^*_{k'k''} a^*_{k'k}\ket{\beta_{k''}} - 2\sum_{k''}b^*_{k''k}\ket{\beta_{k''}}- 2\sum_{k'}a^*_{k'k}\ket{\alpha_{k'}} + \ket{k}, \label{Uk}
\end{equation}
where $D_{k'k''}=\braket{\alpha_{k'}}{\beta_{k''}}$, $k'$ runs from 0 to $m-1$, $k''$ from 0 to $n-1$, and $k$ from 0 to $N-1$.

\begin{figure}[!h]
\centering
\subfigure[fig2][Cycle with $2$-site polygons]{\label{fig2}
\begin{tikzpicture}[scale=0.45]
\draw (0,0) circle (3cm);
\draw [blue,ultra thick] (44:3.7cm) arc (44:100:3.7cm);
\draw [blue,ultra thick] (44:3.7cm) -- (44:2.3cm);
\draw [blue,ultra thick] (100:3.7cm) -- (100:2.3cm);
\draw [blue,ultra thick] (44:2.3cm) arc (44:100:2.3cm);

\draw [blue,ultra thick] (-28:3.7cm) arc (-28:28:3.7cm);
\draw [blue,ultra thick] (-28:3.7cm) -- (-28:2.3cm);
\draw [blue,ultra thick] (28:3.7cm) -- (28:2.3cm);
\draw [blue,ultra thick] (-28:2.3cm) arc (-28:28:2.3cm);

\draw [blue,ultra thick] (-100:3.7cm) arc (-100:-44:3.7cm);
\draw [blue,ultra thick] (-44:3.7cm) -- (-44:2.3cm);
\draw [blue,ultra thick] (-100:3.7cm) -- (-100:2.3cm);
\draw [blue,ultra thick] (-100:2.3cm) arc (-100:-44:2.3cm);

\draw [blue,ultra thick] (-172:3.7cm) arc (-172:-116:3.7cm);
\draw [blue,ultra thick] (-172:3.7cm) -- (-172:2.3cm);
\draw [blue,ultra thick] (-116:3.7cm) -- (-116:2.3cm);
\draw [blue,ultra thick] (-172:2.3cm) arc (-172:-116:2.3cm);

\draw [blue,ultra thick] (116:3.7cm) arc (116:172:3.7cm);
\draw [blue,ultra thick] (116:3.7cm) -- (116:2.3cm);
\draw [blue,ultra thick] (172:3.7cm) -- (172:2.3cm);
\draw [blue,ultra thick] (116:2.3cm) arc (116:172:2.3cm);

\draw [red,dashed,ultra thick] (64:4cm) arc (64:8:4cm);
\draw [red,dashed,ultra thick] (64:4cm) -- (64:2cm);
\draw [red,dashed,ultra thick] (8:4cm) -- (8:2cm);
\draw [red,dashed,ultra thick] (64:2cm) arc (64:8:2cm);

\draw [red,dashed,ultra thick] (-8:4cm) arc (-8:-64:4cm);
\draw [red,dashed,ultra thick] (-8:4cm) -- (-8:2cm);
\draw [red,dashed,ultra thick] (-64:4cm) -- (-64:2cm);
\draw [red,dashed,ultra thick] (-8:2cm) arc (-8:-64:2cm);

\draw [red,dashed,ultra thick] (-80:4cm) arc (-80:-136:4cm);
\draw [red,dashed,ultra thick] (-80:4cm) -- (-80:2cm);
\draw [red,dashed,ultra thick] (-136:4cm) -- (-136:2cm);
\draw [red,dashed,ultra thick] (-80:2cm) arc (-80:-136:2cm);

\draw [red,dashed,ultra thick] (208:4cm) arc (208:152:4cm);
\draw [red,dashed,ultra thick] (-152:4cm) -- (-152:2cm);
\draw [red,dashed,ultra thick] (152:4cm) -- (152:2cm);
\draw [red,dashed,ultra thick] (208:2cm) arc (208:152:2cm);

\draw [red,dashed,ultra thick] (136:4cm) arc (136:80:4cm);
\draw [red,dashed,ultra thick] (136:4cm) -- (136:2cm);
\draw [red,dashed,ultra thick] (80:4cm) -- (80:2cm);
\draw [red,dashed,ultra thick] (136:2cm) arc (136:80:2cm);

\draw[fill] (90:3cm) circle [radius=0.1] node[above] {\scriptsize{$0$}};
\draw[fill] (54:3cm) circle [radius=0.1] node[above] {\scriptsize{$1$}};
\draw[fill] (18:3cm) circle [radius=0.1] node[right] {\scriptsize{$2$}};
\draw[fill] (-18:3cm) circle [radius=0.1] node[right] {\scriptsize{$3$}};
\draw[fill] (-54:3cm) circle [radius=0.1] node[below] {\scriptsize{$4$}};
\draw[fill] (-90:3cm) circle [radius=0.1] node[below] {\scriptsize{$5$}};
\draw[fill] (-126:3cm) circle [radius=0.1] node[below] {\scriptsize{$6$}};
\draw[fill] (-162:3cm) circle [radius=0.1] node[left] {\scriptsize{$7$}};
\draw[fill] (162:3cm) circle [radius=0.1] node[left] {\scriptsize{$8$}};
\draw[fill] (126:3cm) circle [radius=0.1] node[above] {\scriptsize{$9$}};

\end{tikzpicture} }
\subfigure[fig1][Glued triangle tree]{\label{fig1}
\begin{tikzpicture}[scale=0.37]
\draw [red,dashed,ultra thick]  (10,0) -- (5.5,5.5) -- (5.5,-5.5)  -- (10,0);
\draw [blue,ultra thick]  (-10,0) -- (-5.5,5.5) -- (-5.5,-5.5)  -- (-10,0);
\draw [red,dashed,ultra thick]  (-7,4) -- (-2,7) -- (-2,1) --  (-7,4);
\draw [blue,ultra thick]  (7,4) -- (2,7) -- (2,1) --  (7,4);
\draw [red,dashed,ultra thick]  (-7,-4) -- (-2,-7) -- (-2,-1) --  (-7,-4);
\draw [blue,ultra thick]  (7,-4) -- (2,-7) -- (2,-1) --  (7,-4);
\draw [red,dashed,ultra thick]  (4,6) -- (-1,7.5) -- (-1,4.5) --  (4,6);
\draw [red,dashed,ultra thick]  (4,2) -- (-1,3.5) -- (-1,0.5) --  (4,2);
\draw [red,dashed,ultra thick]  (4,-6) -- (-1,-7.5) -- (-1,-4.5) --  (4,-6);
\draw [red,dashed,ultra thick] (4,-2) -- (-1,-3.5) -- (-1,-0.5) --  (4,-2);
\draw [blue,ultra thick]   (-4,6) -- (1,7.5) -- (1,4.5) --  (-4,6);
\draw [blue,ultra thick]   (-4,2) -- (1,3.5) -- (1,0.5) --  (-4,2);
\draw [blue,ultra thick]  (-4,-6) -- (1,-7.5) -- (1,-4.5) --  (-4,-6);
\draw [blue,ultra thick]  (-4,-2) -- (1,-3.5) -- (1,-0.5) --  (-4,-2);
\draw[fill] (0,7) circle [radius=0.1];
\node[below] at (0.4,7) {\scriptsize{$14$}};
\draw (0,7) -- (3,6);
\draw (0,7) -- (-3,6);
\draw[fill] (3,6) circle [radius=0.1];
\node[below] at (3.4,6) {\scriptsize{$18$}};
\draw[fill] (-3,6) circle [radius=0.1];
\node[below] at (-3.3,6) {\scriptsize{$6$}};
\draw (0,5) -- (3,6);
\draw (0,5) -- (-3,6);
\draw[fill] (0,5) circle [radius=0.1];
\node[above] at (-0.4,5) {\scriptsize{$13$}};
\draw (3,6) -- (6,4);
\draw (-3,6) -- (-6,4);
\draw (3,2) -- (6,4);
\draw (-3,2) -- (-6,4);
\draw[fill] (6,4) circle [radius=0.1];
\node[below] at (6.3,3.5) {\scriptsize{$20$}};
\draw[fill] (-6,4) circle [radius=0.1];
\node[below] at (-6.3,3.5) {\scriptsize{$2$}};
\draw (-6,4) -- (-6,-4);
\draw[fill] (0,3) circle [radius=0.1];
\node[below] at (0.4,3) {\scriptsize{$12$}};
\draw (0,3) -- (3,2);
\draw (0,3) -- (-3,2);
\draw[fill] (-3,2) circle [radius=0.1];
\node[above] at (-3.3,2) {\scriptsize{$5$}};
\draw[fill] (3,2) circle [radius=0.1];
\node[above] at (3.4,2) {\scriptsize{$17$}};
\draw (0,1) -- (3,2);
\draw (0,1) -- (-3,2);
\draw[fill] (0,1) circle [radius=0.1];
\node[above] at (-0.4,1) {\scriptsize{$11$}} ;
\draw[fill] (9,0) circle [radius=0.1] node[below] {\scriptsize{$21$}};
\draw[blue,ultra thick] (9,0) circle (1cm); 
\draw (6,4) -- (9,0);
\draw (6,-4) -- (9,0);
\draw[fill] (-9,0) circle [radius=0.1] node[below] {\scriptsize{$0$}};
\draw[red,dashed,ultra thick] (-9,0) circle (1cm); 
\draw (-6,4) -- (-9,0);
\draw (-6,-4) -- (-9,0);
\draw[fill] (0,-1) circle [radius=0.1];
\node[below] at (0.4,-1) {\scriptsize{$10$}};
\draw (0,-1) -- (0,-3);
\draw (0,1) -- (0,3);
\draw (0,-5) -- (0,-7);
\draw (0,5) -- (0,7);
\draw (0,-1) -- (3,-2);
\draw (0,-1) -- (-3,-2);
\draw (3,-2) -- (3,-6);
\draw (3,2) -- (3,6);
\draw[fill] (3,-2) circle [radius=0.1];
\node[below] at (3.4,-2) {\scriptsize{$16$}};
\draw[fill] (-3,-2) circle [radius=0.1];
\node[below] at (-3.3,-2){\scriptsize{$4$}} ;
\draw (-3,-2) -- (-3,-6);
\draw (-3,2) -- (-3,6);
\draw (0,-3) -- (3,-2);
\draw (0,-3) -- (-3,-2);
\draw[fill] (0,-3) circle [radius=0.1];
\node[above] at (-0.3,-3) {\scriptsize{$9$}};
\draw[fill] (6,-4) circle [radius=0.1];
\draw (6,-4) -- (6,4);
\node[above] at (6.3,-3.5) {\scriptsize{$19$}};
\draw[fill] (-6,-4) circle [radius=0.1];
\node[above] at (-6.3,-3.5) {\scriptsize{$1$}};
\draw[fill] (0,-5) circle [radius=0.1];
\node[below] at (0.3,-5) {\scriptsize{$8$}};
\draw (3,-6) -- (6,-4);
\draw (-3,-6) -- (-6,-4);
\draw (3,-2) -- (6,-4);
\draw (-3,-2) -- (-6,-4);
\draw (0,-5) -- (3,-6);
\draw (0,-5) -- (-3,-6);
\draw[fill] (3,-6) circle [radius=0.1];
\node[above] at (3.4,-6) {\scriptsize{$15$}};
\draw[fill] (-3,-6) circle [radius=0.1];
\node[above] at (-3.3,-6) {\scriptsize{$3$}};
\draw (0,-7) -- (3,-6);
\draw (0,-7) -- (-3,-6);
\draw[fill] (0,-7) circle [radius=0.1];
\node[above] at (-0.3,-7) {\scriptsize{$7$}};
\end{tikzpicture}  }

\caption{Examples of tessellations for the  10-cycle and a glued triangle tree with 22 vertices: $(a)$~Depicts tessellations with polygons of size $2$ for the cycle and $(b)$~depicts tessellations with polygons of size $3$ (except for two vertices) for the glued triangle tree.  $U_0$ is associated to the blue tessellation (solid line) and $U_1$ is associated to the red tessellation (dashed line). }
\label{fig:tess}
\end{figure}
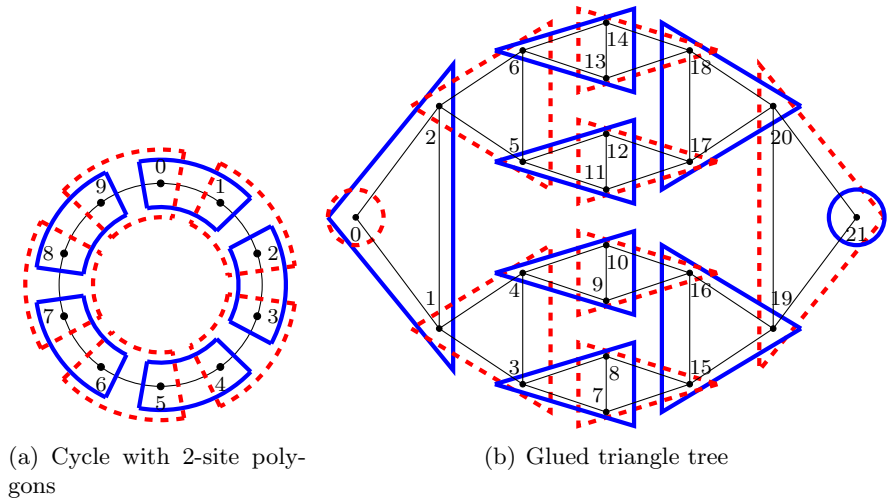

In order to give an alternative form of the definition of staggered QWs, let us define the notion of \textit{orthogonal reflection} of a graph.
\begin{definition}
A unitary and Hermitian operator $U$ is called an orthogonal reflection of a graph if the eigenvectors of $U$ associated with eigenvalue 1 have non-overlapping nonzero entries and the sum of those eigenvectors has no zero entries in the orthonormal basis associated to the vertices of the graph.
\end{definition}
This definition is basis independent because a basis change keeps the orthogonal reflection invariant and changes the orthonormal basis associated to the graph. A staggered QW on a graph $\Gamma$ can be alternatively defined by a unitary evolution operator that is a product of orthogonal reflections of $\Gamma$ that covers the edges of the graph. Given an orthogonal reflection in the basis of ${\cal H}^N$ associated with the vertices of $\Gamma$, it is possible to build a disconnected $N$-graph as a disjoint union of cliques, where each clique is associated to an eigenvector with eigenvalue 1 and the clique is formed with the vertices that have nonzero amplitudes. Each clique defines a polygon and the set of polygons is the tessellation associated with this orthogonal reflection. The product of orthogonal reflections defines the final graph after the union (and merging) of the edges of all cliques. This final graph must be equal to $\Gamma$. The union of the tessellations does not have disposable edges in the sense that every edge is inside a polygon of some tessellation. 

It is tempting to remove the requirement that the polygons must be cliques. The removal of this requirement is troublesome because after one application of the unitary operator $U_0$ (or $U_1$) with non-clique polygons, the walker can go to non-adjacent vertices. In quantum walk models, we expect that the walker only hops to adjacent vertices, and successive applications of shift operators move the walker to distant locations. Our definition does not permit the tessellations used by Falk in Ref.~\cite{Falk:2013}, which were also used by Ambainis \textit{et al}. in Ref.~\cite{Ambainis:2013} and  Patel  \textit{et al}. in Refs.~\cite{Patel05,Patel:2010}. Those papers go against the requirement described in Ref.~\cite{Aharonov:2000}, which states that a general form of quantum walks must respect the structure of the graph.   The tessellations used in those references in fact refer to the graph depicted in Fig.~\ref{fig:grid}. For the two-dimensional lattice (with degree 4), it is necessary to use at least four tessellations because every vertex belongs to four maximum cliques of size two.

\begin{figure}[!h]
\centering
\begin{tikzpicture}[scale=0.45]
\draw (-1,0) -- (12,0);
\draw (-1,3) -- (12,3);
\draw (-1,6) -- (12,6);
\draw (-1,9) -- (12,9);
\draw (1,-2) -- (1,11);
\draw (4,-2) -- (4,11);
\draw (7,-2) -- (7,11);
\draw (10,-2) -- (10,11);
\draw [dashed,red,ultra thick] (-1,1) -- (2,1) -- (2,-2);
\draw [dashed,red,ultra thick] (-1,2) -- (2,2) -- (2,7) -- (-1,7);
\draw [dashed,red,ultra thick] (3,11) -- (3,8) -- (8,8) -- (8,11);
\draw [dashed,red,ultra thick] (-1,8) -- (2,8) -- (2,11);
\draw [dashed,red,ultra thick] (9,11) -- (9,8) -- (12,8);
\draw [dashed,red,ultra thick] (12,7) -- (9,7) -- (9,2) -- (12,2);
\draw [dashed,red,ultra thick] (3,-2) -- (3,1) -- (8,1) -- (8,-2);
\draw [dashed,red,ultra thick] (12,1) -- (9,1) -- (9,-2);
\draw [dashed,red,ultra thick] (3,2) rectangle (8,7);
\draw [blue,ultra thick] (0,-1) rectangle (5,4);
\draw [blue,ultra thick] (6,-1) rectangle (11,4);
\draw [blue,ultra thick] (6,5) rectangle (11,10);
\draw [blue,ultra thick] (0,5) rectangle (5,10);
\draw[fill] (1,0) circle [radius=0.1] node[below] {\scriptsize{$00$}};
\draw[fill] (4,0) circle [radius=0.1] node[below] {\scriptsize{$10$}};
\draw[fill] (7,0) circle [radius=0.1] node[below] {\scriptsize{$20$}};
\draw[fill] (10,0) circle [radius=0.1] node[below] {\scriptsize{$30$}};
\draw[fill] (1,3) circle [radius=0.1] node[below] {\scriptsize{$01$}};
\draw[fill] (4,3) circle [radius=0.1] node[below] {\scriptsize{$11$}};
\draw[fill] (7,3) circle [radius=0.1] node[below] {\scriptsize{$21$}};
\draw[fill] (10,3) circle [radius=0.1] node[below] {\scriptsize{$31$}};
\draw[fill] (1,6) circle [radius=0.1] node[below] {\scriptsize{$02$}};
\draw[fill] (4,6) circle [radius=0.1] node[below] {\scriptsize{$12$}};
\draw[fill] (7,6) circle [radius=0.1] node[below] {\scriptsize{$22$}};
\draw[fill] (10,6) circle [radius=0.1] node[below] {\scriptsize{$32$}};
\draw[fill] (1,9) circle [radius=0.1] node[below] {\scriptsize{$03$}};
\draw[fill] (4,9) circle [radius=0.1] node[below] {\scriptsize{$13$}};
\draw[fill] (7,9) circle [radius=0.1] node[below] {\scriptsize{$23$}};
\draw[fill] (10,9) circle [radius=0.1] node[below] {\scriptsize{$33$}};

\draw (5,11) -- (12,4);
\draw (-1,11) -- (12,-2);
\draw (-1,5) -- (6,-2);

\draw (6,11) -- (-1,4);
\draw (12,11) -- (-1,-2);
\draw (12,5) -- (5,-2);
\end{tikzpicture}
\caption{Example of tessellations with polygons of size $4$ on a non-planar regular graph of degree~6. $U_0$ is associated to the blue tessellation (solid line) and $U_1$ is associated to the red tessellation (dashed line). }
\label{fig:grid}
\end{figure}
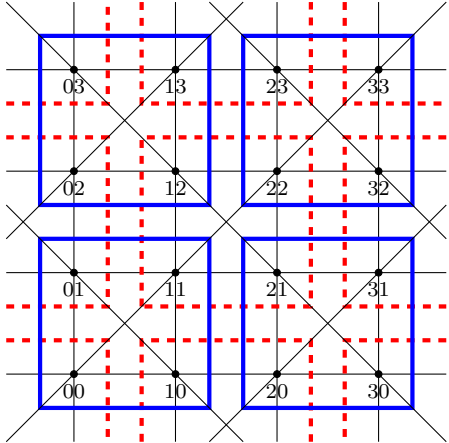

It is also tempting to remove the requirement that all edges must be inside of a polygon. The removal of this requirement is troublesome because it generates a logical indeterminacy, similar to one in the discussion about the use of non-clique polygons in the previous paragraph. If an edge is not inside a polygon, it means that it can be removed generating no change in the dynamics. In this case non-isomorphic graphs would have exactly the same dynamics and would lead to confusion and logical problems.

\subsection{Generalized Staggered Quantum Walks}\label{sec:ExtStaggered}

Usually quantum walk models are defined with no marked vertices. In this case the dynamics has no special vertex and the amplitude associated with some specific vertex will not increase above average unless the initial condition is a special one. For instance, the coined model defined in Ref.~\cite{Aharonov:2000} has those characteristics. The standard definition must be changed in order to mark a vertex. In the coined case, the coins for the marked vertices are different from the coins for non-marked vertices~\cite{Shenvi:2003}.
  
We can extend the staggered quantum walk definition by removing the requirement that each tessellation must cover the entire graph. It is allowed to use partial tessellations so that the partial-tessellation union covers all vertices (not necessarily all edges). The vertices that are not inside all tessellations are the marked ones. An equivalent way of obtaining the same extension is by allowing zero amplitudes in the entries of polygons $\alpha_k$ and $\beta_k$. 
The goal of this generalization is to define search algorithms. For example, Grover's algorithm~\cite{Grover:1997a} can be seen as a generalized staggered quantum walk on the complete graph in the following way: Tessellation $\alpha$ has polygons of size one for each non-marked vertex and tessellation $\beta$ has only one polygon with all vertices. The vertices that do not belong to tessellation $\alpha$ are the marked ones. The evolution operator of this generalized staggered model is equal to the one used in Grover's algorithm.

\section{Szegedy's Quantum Walks}\label{sec:Szegedy}


In this section we review Szegedy's QW model with the goal of finding the connection with the staggered QW model. Consider a connected bipartite graph $\Gamma(X,Y,E)$, where $X,Y$ are disjoint sets of vertices and $E$ is the set of non-directed edges. Let 
\begin{equation}
		\left(\begin{array}[]{cc}
		  0 & A \\
			A^T & 0
	\end{array}\right)
\end{equation}
be the biadjacency matrix of  $\Gamma(X,Y,E)$. Using $A$, define $P$ as a probabilistic map from $X$ to $Y$. Using $A^T$, define $Q$  as a probabilistic map from $Y$ to $X$. If $P$ is an $m\times n$ matrix, $Q$ will be an $n\times m$ matrix, both are right-stochastic with the property that each row sums to 1.

To define a quantum walk on the bipartite graph, we associate it with a Hilbert space ${\cal H}^{m n} = {\cal H}^{m}\otimes {\cal H}^{n} $, where $ m = | X |$ and $n = | Y | $, the computational basis of which is $ \big \{\ket {x, y}: x \in X, y \in Y \big \} $.
Szegedy's quantum walk is defined by the evolution operator
\begin{equation}\label{ht_U_ev}
    W \,=\, R_1 \, R_0,
\end{equation}
where
\begin{eqnarray}
  R_0 &=& 2\sum_x \ket{\phi_x}\bra{\phi_x} - I, \label{ht_RA}\\
  R_1 &=& 2\sum_y\ket{\psi_y}\bra{\psi_y} - I, \label{ht_RB}
\end{eqnarray}
and
\begin{eqnarray}
  \ket{\phi_x} &=&  \sum_{y\in Y} \sqrt{p_{x y}} \, \ket{x,y}, \label{ht_phi_x} \\
  \ket{\psi_y}  &=&  \sum_{x\in X} \sqrt{q_{y x}} \, \ket{x,y}. \label{ht_psi_y}
\end{eqnarray}
It is straightforward to show that
\begin{eqnarray}
W \ket{x,y} &=& 4 \sqrt{p_{x y}}\sum_{y'\in Y} C_{y'x} \ket{\psi_{y'}} - 2\sqrt{q_{y x}}\ket{\psi_{y}}- 2 \sqrt{p_{x y}}\ket{\phi_{x}} + \ket{x,y}, \label{Wketxy}
\end{eqnarray}
where $C_{yx}=\braket{\psi_y}{\phi_x}$.  For a detailed review see~\cite{Portugal:Book}.

\section{Equivalence between Szegedy's and Staggered QWs}\label{sec:Szegedisstag}

In this section we show that any instance of Szegedy's model is equivalent to a staggered quantum walk version. The converse is not true. We show that some instances of the staggered quantum walk model are equivalent to Szegedy's model and we characterize when staggered walks cannot be cast into Szegedy's framework.

\subsection{Szegedy's Walks are Staggered Walks!}\label{subsec:Szegedisstag}

Although Szegedy's model is defined on the Hilbert space ${\cal H}^{mn}$ associated with a bipartite graph $\Gamma(X,Y,E)$, the dynamics takes place in the subspace spanned by the edges of $\Gamma$. $W$ acts trivially on vectors $\ket{x,y}$ that do not belong to the bipartite graph and the initial condition does not include those edges. In this section we show how to define a staggered QW based on generic stochastic matrices $P$ and $Q$ equivalent to Szegedy's QW.

Let $N$ be the number of edges of $\Gamma$. Define two sets of polygons in the line graph $L(\Gamma)$ by
\begin{eqnarray}
  \ket{\alpha_x} &=&  \sum_{y\in Y} \sqrt{p_{x y}} \, \ket{f(x,y)}, \label{new_alpha_k} \label{eq:alpha_x}\\
	\ket{\beta_y} &=&  \sum_{x\in X} \sqrt{q_{y x}} \, \ket{f(x,y)}, \label{new_beta_k}\label{eq:beta_y}
\end{eqnarray}
where $f$ is a bijection between the edge set $E$ of $\Gamma$ and the labels $k$ of $L(\Gamma)$, which has associated the Hilbert space ${\cal H}^{N}$ with the computational basis $\big\{\ket{k}$, $k=0,...,N-1\big\}$. Each polygon ${\alpha_x}$ is a clique and the union of polygons $\alpha_x$ tessellates $L(\Gamma)$. The same is true for polygons ${\beta_y}$.

Suppose that $f(x,y)=k$, where $\{x,y\}\in E$. Then
\begin{equation}
	U\ket{k}=\left(2\sum_{y'=0}^{n-1} \ket{\beta_{y'}}\bra{\beta_{y'}} - I\right)\left(2\sum_{x'=0}^{m-1} \ket{\alpha_{x'}}\bra{\alpha_{x'}} - I\right)\ket{f(x,y)}.
\end{equation}
Using that $\braket{\alpha_{x'}}{f(x,y)}=\sqrt{p_{xy}}\delta_{xx'}$ and $\braket{\alpha_{y'}}{f(x,y)}=\sqrt{q_{yx}}\delta_{yy'}$, we obtain
\begin{equation}
U\ket{k} = 4 \sqrt{p_{x y}}\sum_{y'} \braket{\beta_{y'}}{\alpha_x} \ket{\beta_{y'}} - 2\sqrt{q_{y x}}\ket{\beta_{y}}- 2 \sqrt{p_{x y}}\ket{\alpha_{x}} + \ket{f(x,y)}. \label{Uketk}
\end{equation}
Eqs.~(\ref{Uketk}) and~(\ref{Wketxy}) are essentially the same results if we apply bijection $f$ to the labels $(x,y)$ of the kets $\ket{x,y}$ of Eq.~(\ref{Wketxy}). This proves 

\begin{proposition}\label{prop:Sz=Falk}
Any instance of Szegedy's quantum walk model on a bipartite graph $\Gamma$ is equivalent to a staggered quantum walk version on the line graph of $\Gamma$.
\end{proposition}

\subsection{Which Staggered QWs are instances of Szegedy's Model?}\label{sec:FalktoSz}

Consider a staggered quantum walk on a $N$-graph $\Gamma'$ with the following restriction: The intersections of polygons belonging to different tessellations contain one vertex. If polygons $\alpha_{k}$ and $\beta_{k'}$ share one vertex, the ket label of this vertex in Szegedy's model will be $\ket{k,k'}$. Because all vertices must be in both tessellations, this labeling method establishes a bijection $f'$ between the computational basis of ${\cal H}^{N}$ used in the staggered model and the set of labels of the form $\ket{k,k'}$, that will be used to build an instance of Szegedy's model on a bipartite graph $\Gamma$, the line graph of which is $\Gamma'$.\footnote{Given a line graph $\Gamma'$, there is only one bipartite graph $\Gamma$ such that $L(\Gamma)=\Gamma'$~\cite{Whitney:1932}.} By construction, the set of labels of the form  $\ket{k,k'}$ belonging to polygon $\alpha_{k}$ shares the same value of $k$ with different values of $k'$. After applying the bijection $f'$ to the kets of polygon  $\alpha_{k}$, we obtain a vector $\ket{\tilde\alpha_k}$ in ${\cal H}^{m n}$ ($m,n$ are the number of polygons in tessellations $\alpha,\beta$) given by
\begin{eqnarray}
  \ket{\tilde\alpha_k} &=&  \sum_{k'} a_{k,k'} \ket{k,k'}, \label{newalpha_k} 
\end{eqnarray}
with the same amplitudes of the original polygon  $\alpha_{k}$, where $k'$ runs over the subindices of polygons $\beta_{k'}$ that have overlap with $\alpha_{k}$, that is, $\beta_{k'}\cap\alpha_k\neq\emptyset$. Analogously, after applying the bijection $f'$ to the kets of polygon  $\beta_{k'}$, we obtain a vector $\ket{\tilde\beta_{k'}}$ in ${\cal H}^{m n}$ which has same value of $k'$ as in the second slot of the kets, and it is given by
\begin{eqnarray}
  \ket{\tilde\beta_{k'}} &=&  \sum_{k} b_{k',k} \ket{k,k'}, \label{newbeta_k}
\end{eqnarray}
where $k$ runs over the subindices of polygons $\alpha_{k}$ that have overlap with $\beta_{k'}$, that is, $\alpha_k\cap\beta_{k'}\neq\emptyset$.

We must consider a second restriction regarding the amplitudes $a_{k,k'}$ and $b_{k',k}$ of the staggered model. We assume that $a_{k,k'} $ and $b_{k',k}$ are non-negative real numbers obeying the constraints $\sum_{k'} a_{k,k'}^2 = 1$ and $\sum_{k} b_{k',k}^2=1$, that is, $\ket{\tilde\alpha_{k}}$ and $\ket{\tilde\beta_{k'}}$ are real unit vectors.
After imposing this restriction, such instances of the staggered quantum walk
model can be put into the standard Szegedy's framework. In fact, the evolution driven by $U$, which is defined using $\ket{\alpha_{k}}$ and $\ket{\beta_{k}}$ with the restricted version of the amplitudes, is equivalent to the evolution driven by $W$ using $\ket{\tilde\alpha_{k}}$ and $\ket{\tilde\beta_{k'}}$ instead of $\ket{\phi_{x}}$ and $\ket{\psi_{y}}$. The equivalence is brought out by applying bijection $f'$ to the labels of the kets of $U\ket{k}$ given by Eq.~(\ref{Uk}) and comparing with $W\ket{f'(k)}$. The second restriction (real amplitudes) can be eliminated by a straightforward extension of Szegedy's model, which is addressed in the next section.

The instances of the staggered quantum walk models that cannot be cast into Szegedy's framework with the above method are those that have at least one overlapping polygons of tessellations $\alpha$ and $\beta$ with more than one vertex. To prove this fact, let $\ket{j_1}$ and $\ket{j_2}$ be kets associated with two vertices in the overlap between polygons $\alpha_{k}$ and $\beta_{k'}$. There is a contradiction because $\ket{j_1}$ and $\ket{j_2}$ must be mapped to kets of the form $\ket{k,k'}$ both sharing the same label $k$ and $k'$. In fact, in order to build bijection $f'$, we need to map $\ket{j_1}$ to $\ket{k,k'_1}\in {\cal H}^{m n}$ and $\ket{j_2}$ to $\ket{k,k'_2}\in {\cal H}^{m n}$ sharing the same $k$ because both belong to polygon $\alpha_{k}$. We have to take $k'_1\neq k'_2$ because different vertices must have different labels. Besides, $\ket{j_1}$ and  $\ket{j_2}$ belong also to polygon $\beta_{k'}$, then labels $k'_1$ and $k'_2$ must be the same, which is a contradiction. 
Therefore, there is no bijection between the computational basis of ${\cal H}^{N}$ and the set of kets of the form $\ket{k,k'}$ with the property of sharing the same $k$ for polygons of tessellation $\alpha$ and simultaneously sharing the same $k'$ for polygons of tessellation $\beta$. 


The simplest non-trivial examples of staggered QWs that cannot be cast into Szegedy's framework are obtained on the graph with only two connected vertices ($N=2$). Take tessellations $\ket{\alpha}=\big(\ket{0}+\ket{1}\big)/\sqrt 2$ and $\ket{\beta}=\ket{\psi}$, where 
$$\ket{\psi}=\cos\frac{\theta}{2}\,\ket{0}+\sin\frac{\theta}{2}\,\ket{1},$$ 
for some angle $\theta$ that is not a rational multiple of $\pi$. The evolution operator is
$$U=\left[ \begin {array}{cc} \sin  \theta  &\cos \theta \\ 
\noalign{\medskip}-\cos \theta  &\sin \theta \end {array} \right],$$
which is non-trivial because there is no integer $r$ such that $U^r=I$. On the other hand, any Szegedy's quantum walk $W$ on any graph with two edges is trivial because $W^2=I$. Notice that we do not need to go back to Szegedy's framework to prove this result because we can use the equivalent staggered versions. This proves 

\begin{proposition}\label{prop:exampleofconverse}
There are instances of the staggered quantum walk model with two tessellations that cannot be cast into Szegedy's quantum walk framework. Those instances have at least one polygon intersection with at least two vertices in common.
\end{proposition}

For example, the glued triangle tree of Fig.~\ref{fig:tess} cannot be cast into Szegedy's framework because it is not a line graph. Even if it was the line graph of a bipartite graph, it would not be possible to obtain an equivalent Szegedy's QW because there are polygons intersections with two vertices in common.

An alternative proof of the above proposition in terms of graph theory is as follows. If the quantum walk graph $\Gamma$ is not a line graph, by inspection we verify that the nine forbidden Beineke induced subgraphs~\cite{Bei70} require either tessellations with two vertices in common or at least three-tessellations. If graph $\Gamma$ is a line graph, then it has a Krausz partition~\cite{Kra43}. If the Krausz partition is two-colorable, $\Gamma$ is a line graph of a bipartite graph and can be reduced to Szegedy's framework by using Eqs.~(\ref{newalpha_k}) and~(\ref{newbeta_k}). If the Krausz partition is not two-collorable, $\Gamma$ requires either tessellations with two vertices in common or at least three-tessellations.

\subsection{Extension of Szegedy's Model}

Szegedy's model can be easily extended by allowing phases in vectors
\begin{eqnarray}
  \ket{\phi_x} &=&  \sum_{y\in Y} \sqrt{p_{x y}}\,\textrm{e}^{i\theta_{xy}} \, \ket{x,y}, \label{newht_phi_x} \\
  \ket{\psi_y}  &=&  \sum_{x\in X} \sqrt{q_{y x}}\,\textrm{e}^{i\theta'_{xy}} \, \ket{x,y}. \label{newht_psi_y}
\end{eqnarray}
Entries $p_{x y},q_{y x}$ are non-negative real numbers obeying the constraints $\sum_y p_{x y}=\sum_x q_{y x}=1$, as in the original model. It is straightforward to show that $W$ is unitary in the extended case. Using the same algorithm of Sec.~\ref{sec:FalktoSz}, we conclude that any staggered quantum walk with the property that the intersections of polygons belonging to different tessellations contain one vertex is equivalent to an instance of the extended version of Szegedy's model. This proves 
\begin{proposition}\label{prop:someFalk=Sz}
Any instance of the staggered quantum walk model with two tessellations which has only one vertex in common in each polygon intersection can be cast into the extended version of Szegedy's quantum walk framework.
\end{proposition}
Notice that a staggered quantum walk model on a graph $\Gamma$ with two tessellations which have only one vertex in common in each polygon intersection has a two-colorable Krauz partition. Therefore  $\Gamma$ is a line graph of a bipartite graph.

\subsection{Searching in Szegedy's Model}

Szegedy's model can be used to search for a vertex in a set of marked vertices in some graph $\Gamma$. The strategy is an extension of the classical method, which uses random walks on the same graph. The key concept in the classical case is the \textit{hitting time}, which is the average time to hit a marked vertex for the first time using a random walk with stochastic (or transition) matrix $P$ after specifying some initial condition. There is a clever way to find the classical hitting time using the following idea: Convert the original graph $\Gamma$ into a new directed graph $\Gamma'$ (with a new stochastic matrix $P'$) by removing the edges that leave the marked vertices. Marked vertices are converted into sinks. The hitting time obtained using $P'$ is the same using $P$ because as soon as the walker hits a marked vertex using some edge that comes from a non-marked vertex, the walker needs not to go ahead. The advantage of using $P'$ is that it does not have eigenvectors associated with eigenvalue 1 (the spectral norm is smaller than 1), that is, $I-P'$ is an invertible matrix. We need to invert $I-P'$ to find the hitting time if we want to avoid the calculation of the limiting probability distribution. 

Szegedy proposed a quantum version of this procedure. Let $\Gamma(X,E)$ be the original graph, where $X$ is the set of vertices and $E$ the set of edges. Define $\Gamma(X,X',E')$ as a bipartite graph obtained from  $\Gamma(X,E)$ by duplicating $X$ and by converting edges $\{x_i,x_j\}\in E$ into $\{x_i,x'_j\}\in E'$. Until this point, no vertex has been marked and the quantum walk described in Sec.~\ref{sec:Szegedy} can be used taking $P=Q$. As before, define $\Gamma'(X,X',E')$ as a directed bipartite graph by removing the edges of $\Gamma(X,X',E')$ that leaves the marked vertices of $X$ and $X'$. Add new non-directed edges connecting a marked $x$ with its corresponding copy $x'$ for all marked vertices. The quantum walk described in Sec.~\ref{sec:Szegedy} can be used taking $P'=Q'$, where $P'$ and $Q'$ are the new stochastic matrices of $\Gamma'(X,X',E')$. With this framework and taking 
\begin{equation}
\ket{\psi_0}=\frac{1}{\sqrt n}\sum_{x y} \sqrt{p_{x y}} \ket{x,y}
\end{equation}
as initial condition in ${\cal H}^{n}\otimes {\cal H}^{n}$,  Szegedy showed that the detection problem on $\Gamma'(X,X',E')$ can be solved with a quadratic speedup compared to the time complexity of the same problem using random walks with symmetric and ergodic stochastic matrix $P$ in $\Gamma(X,E)$~\cite{Szegedy:2004}.

The method of converting Szegedy's QWs into staggered QWs described in Sec.~\ref{subsec:Szegedisstag} can be used in the searching context after some minor modifications. Let $L(\Gamma)$ be the line graph of $\Gamma(X,X',E')$ and define polygons $\alpha_x$ and $\beta_y$ as the ones in Eqs.~(\ref{eq:alpha_x}) and~(\ref{eq:beta_y}) for all $x$ and $y$ that do not belong to the set of marked vertices. Those sets of polygons define partial tessellations of $L(\Gamma)$ which is used to define a generalized staggered QW equivalent to Szegedy's QW on $\Gamma'(X,X',E')$. The initial condition in ${\cal H}^{N}$ must be
\begin{equation}
\ket{\psi_0}=\frac{1}{\sqrt n}\sum_{x y} \sqrt{p_{x y}} \ket{f(x,y)},
\end{equation}
where $f$ is the bijection from the $E$ to the labels of $L(\Gamma)$. 

The last ingredient in Szegedy's QW is the measurement, which is performed in the computational basis of $X$. At this moment, the copy $X'$ is traced out. If the measurement is performed at the right time, which is the \textit{quantum hitting time}, the probability of finding a marked vertex will be high enough. The version using generalized staggered QWs is obtained as follows. For $x\in X$ define projectors 
\begin{equation}
\Pi_x=\sum_{x'\in \alpha_x}\ket{x'}\bra{x'}.
\end{equation}
Notice that $\sum_x \Pi_x=I$. By measuring observable
\begin{equation}
{\cal O}=\sum_{x\in X} x \Pi_x,
\end{equation}
we obtain some value $x_0\in X$, which is expected to be one of the marked vertices if the measurement is performed at the hitting time. In the staggered version, we are looking for the missing polygons of the (partial) $\alpha$-tessellation. This proves

\begin{proposition}\label{prop:searching}
Szegedy's searching framework on a bipartite graph $\Gamma$ can be simulated by instances of generalized staggered quantum walks on the line graph of $\Gamma$.
\end{proposition}


\section{Hitting Time on Staggered Walks}\label{sec:HT}

We have showed that any instance of the extended Szegedy's model can be cast
into an equivalent staggered quantum walk. The converse statement is not true. There are instances of staggered walks that are not equivalent to  
any instance of the extended version of Szegedy's model. We take advantage of the similarity of both models to define a notion of hitting time on the staggered model, which is specially useful on those staggered walks that cannot be cast into Szegedy's framework.

Let $M$ be the set of marked vertices. Inspired again by Falk's paper~\cite{Falk:2013}, the evolution operator $U_M$ for searching a marked vertex is defined by
\begin{equation}\label{URM}
	U_M = R_M U_1 R_M U_0,
\end{equation}
where
\begin{equation}
	 R_M = 2\sum_{m\in M} \ket{m}\bra{m} - I.
\end{equation}
$R_M$ is a orthogonal reflection and flips non-marked vertices.

\begin{definition}\label{defHT}
The quantum hitting time $H_M$ of a staggered quantum walk starting from the initial condition $\ket{\psi(0)}$ is the smallest number of steps $T$ such that
\begin{equation}\label{FT}
	\frac{1}{T+1}\sum_{t=0}^{T}\left\|\, \ket{\psi(t)}-\ket{\psi(0)}\, \right\|^2\ge 1-\frac{|M|}{N},
\end{equation}
where $\ket{\psi(t)}=(U_M)^t\ket{\psi(0)}$ and $N$ is the number of vertices. 
\end{definition}
The hitting time depends on the initial condition. An important one in this context is the uniform distribution because the value $1-|M|/N$ is the distance between the uniform distribution over all vertices and the uniform distribution over the marked vertices.

\section{Searching using Staggered Walks}\label{sec:SearchusingStag}

$R_M$ can be interpreted as an operator associated with a partial tessellation with $|M|$ polygons each one with one marked vertex. There is a simpler way to analyze the evolution operator (\ref{URM}) by using the fact that $R_M U_1 R_M$ is an orthogonal reflection.

\begin{proposition}\label{prop:UM}
The evolution operator $U_M$ for searching a marked vertex using a staggered quantum walk is a product of two orthogonal reflections.
\end{proposition}
\begin{proof}
Define $U'_1=R_M U_1 R_M$. Using Eq.~(\ref{U_1}), we obtain
\begin{equation}\label{Up1app}
  U'_1= 2\sum_{k=0}^{n-1} \ket{\beta'_k}\bra{\beta'_k} - I,
\end{equation}
where $\ket{\beta'_k}=R_M\ket{\beta_k}$. Using Eq.~(\ref{beta_k}), we obtain
\begin{equation}\label{betapk}
  \ket{\beta'_k}= 
\begin{cases}
    -\ket{\beta_k},& \textrm{if }  \beta_k \cap M = \emptyset\\
    \sum_{k'=0}^{n-1} (-1)^{f(k')} b_{k k'}\ket{k'},& \text{otherwise,}
\end{cases}
\end{equation}
where $f(k)=0$ if $k\in M$ and $f(k)=1$ otherwise. Therefore
$U_M=U'_1 \, U_0$ and using Eqs.~(\ref{Up1app}) and~(\ref{betapk}) it is straightforward to verify that $U'_1$ is an orthogonal reflection. 
$\blacksquare$\end{proof}

Next proposition shows that if we mark all vertices inside some polygons of tessellation $\alpha$ or $\beta$, the walker is not able to find a marked polygon.

\begin{proposition}\label{prop:UM=U}
If the set of marked vertices is a union of polygons of tessellation $\beta$, then $U_M=U$. If the set of marked vertices is a union of polygons of tessellation $\alpha$, then $U_M$ is similar to $U$.
\end{proposition}
\begin{proof} 
If the set of marked vertices is a union of polygons of tessellation $\beta$, Eq.~(\ref{betapk}) reduces to 
\begin{equation}\label{newbetapk}
  \ket{\beta'_k}= 
\begin{cases}
    -\ket{\beta_k},& \textrm{if }  \beta_k \cap M = \emptyset\\
    \ket{\beta_k},& \text{otherwise.}
\end{cases}
\end{equation}
Substituting into Eq.~(\ref{Up1app}), we obtain that $U'_1=U_1$. Therefore $U_M=U$.

The proof of the second statement is analogous. Define $U'_0=R_M U_0 R_M$. Using Eq.~(\ref{U_0}), we obtain
\begin{equation}\label{Up0app}
  U'_0= 2\sum_{k=0}^{m-1} \ket{\alpha'_k}\bra{\alpha'_k} - I,
\end{equation}
where $\ket{\alpha'_k}=R_M\ket{\alpha_k}$. If the set of marked vertices is a union of polygons of tessellation $\alpha$, then using Eq.~(\ref{alpha_k}) we obtain
\begin{equation}\label{alphapk}
  \ket{\alpha'_k}= 
\begin{cases}
    -\ket{\alpha_k},& \textrm{if }  \alpha_k \cap M = \emptyset\\
    \ket{\alpha_k},& \text{otherwise.}
\end{cases}
\end{equation}
Therefore $U'_0=U_0$ and
$U_M=R_M\,U\, R_M$. 
$\blacksquare$\end{proof}

Useful properties of the QW dynamics are obtained from the spectrum of $U_M$. The eigenvalues and eigenvectors can be obtained from the singular values and vectors of the discriminant (or Gram-like) matrix $D$ the entries of which are
\begin{equation}
 D_{k\,k'}=\braket{\alpha_{k}}{\beta'_{k'}},
\end{equation}
where $\ket{\alpha_{k}}$ is given by Eq.~(\ref{alpha_k}) and $\ket{\beta'_{k'}}$ is given by Eq.~(\ref{betapk}). Let $\ket{\nu_j}$ and $\ket{\mu_j}$ be the right and left singular vectors of $D$ respectively and let $\cos\theta_j$ be the corresponding singular values obeying $D\ket{\nu_j}=\cos\theta_j\ket{\mu_j}$ and  $D^\dagger\ket{\mu_j}=\cos\theta_j\ket{\nu_j}$ and $0\le \theta_j\le \pi/2$. Suppose that $\theta_1,...,\theta_s$ are the angles in the open interval $(0,\pi/2)$, where $s$ is the number of singular values of $D$ in the open interval $(0,1)$. Define
\begin{eqnarray}
 A &=& \sum_{k=0}^{m-1} \ket{\alpha_k}\bra{k}, \\
 B &=& \sum_{k'=0}^{n-1} \ket{\beta'_{k'}}\bra{k'}. 
\end{eqnarray}
Let ${\cal A},{\cal B}$ be the vector spaces spanned by vectors $\ket{\alpha_{k}}$ and vectors $\ket{\beta'_{k'}}$ respectively. We reproduce a theorem from~\cite{Szegedy:2004} adapted to our context.

\begin{theorem}\label{theo:Eigenvectors}
The spectrum of $U_M$ obeys:
\begin{enumerate}
\item The eigenvalues of $U_M$ with nonzero imaginary part are exactly $\textrm{e}^{\pm 2i\theta_j}$ for $j=1,...,s$. The corresponding normalized eigenvectors are
\begin{equation}
\frac{A\ket{\mu_j}-\textrm{e}^{\pm i\theta_j}B\ket{\nu_j}}{\sqrt 2 \sin\theta_j}.
\end{equation}
\item ${\cal A}\cap {\cal B}^\perp\,+\,{\cal A}^\perp\cap {\cal B}$ is the $-1$ eigenspace of $U_M$. ${\cal A}\cap {\cal B}^\perp$ is spanned by the left singular vectors of $D$ with singular value 0. ${\cal A}^\perp\cap {\cal B}$ is spanned by the right singular vectors of $D$ with singular value 0.  
\item ${\cal A}\cap {\cal B}\,+\,{\cal A}^\perp\cap {\cal B}^\perp$ is the $+1$ eigenspace of $U_M$. ${\cal A}\cap {\cal B}$ is spanned by the left (or right) singular vectors of $D$ with singular value 1. 
\end{enumerate}
\end{theorem}

The above theorem can be used to find most eigenvectors of the evolution operator when the singular values and vectors of the discriminant matrix can be explicitly found. The eigenvectors that span ${\cal A}^\perp\cap {\cal B}^\perp$ are not described in Theorem~\ref{theo:Eigenvectors}. The staggered Fourier transform is an alternative method to find the spectrum of $U_M$, but it depends on translation symmetries of the tessellating process. The staggered Fourier transform was used in Refs.~\cite{Ambainis:2013,PBF15,Santos:2015}.

\section{Conclusions and Discussions}\label{sec:Conc}

In this paper we have formally defined the staggered QW model on graphs. This model can be obtained by tessellating the graph using polygons obeying locality restrictions. The polygons are cliques and each tessellation covers all vertices. All edges must be inside a polygon of some tessellation. An alternative equivalent definition is based on the notion of orthogonal reflections in the basis associated to the graph, which are unitary and Hermitian operators with the following property: The nonzero amplitudes of the eigenvectors associated with eigenvalue 1 do not have overlap and the sum of those eigenvectors does not have zero entries. The evolution operator of staggered QWs must be a product of orthogonal reflections that covers all edges of the graph.

We have showed that any instance of Szegedy's QW model on a bipartite graph $\Gamma$ is equivalent to an instance of the staggered QW model on the line graph of $\Gamma$ with a special kind of tessellations: There is only one vertex in each polygon intersection and the tessellation union is a Krausz partition. On the other hand, there are instances of the staggered model with two tessellations that cannot be cast into Szegedy's framework, which are the ones with two or more vertices in common in at least one polygon intersection. This characterization shows that Szegedy's model is a restricted version of QWs, which is defined only on line graphs of bipartite graphs, a subset of perfect graphs. This fact helped researchers to obtain generic results using Szegedy's model, different from the results about the coined version, which is usually found in specific graphs. One of the main results of Szegedy's model is that the quantum hitting time is quadratically smaller compared with the classical hitting time on symmetric and ergodic Markov chains. Notice that the classical random walk is defined on a graph $\Gamma$ while the quantum walk is defined on the line graph of the bipartite graph associated to $\Gamma$. When the expressions for the hitting time on different graphs are compared, they have the same structure and it is possible to reach the result. 

There are two main forms of defining searching algorithms in the staggered model: (1) Searching \textit{\`a la} Szegedy is the method that uses partial tessellations in the generalized staggered model. The marked vertices are the ones in the missing polygons. (2) Searching \textit{\`a la} Falk uses the standard tessellation procedure but interlaces the orthogonal reflections with a reflection around the marked vertices. In both cases we can use Szegedy's spectral theorem to calculate the spectrum of the evolution operator when using two tessellations. We have also defined the quantum hitting time on staggered QWs using reflections around the marked vertices which generalizes Szegedy's original definition. We have shown that the evolution operator, originally defined as a product of four reflections, can be written as a product of two reflections.

Because all vertices of the two-dimensional lattice belong to the intersection of four maximum cliques, it is necessary to employ at least four tessellations. The case with exactly four tessellations has one vertex in the polygon intersections but cannot be cast into Szegedy's framework because it employs more than two tessellations. Future works on this issue are important to establish the properties of staggered QWs on two-dimensional lattices, the analysis of which is missing in literature. 

An interesting subject is the equivalence between the staggered model and the coined quantum walk model. This has been established for the line and cycle~\cite{HKS05,PBF15}. A systematic analysis on the equivalence of those models on generic graphs is an important issue for future works.


\section*{Acknowledgements}
 RP acknowledges financial support from Faperj (grant n.~E-26/102.350/2013) and CNPq (grants n.~304709/2011-5, 4741\-43/2013-9, and 400216/2014-0). RAMS acknowledges financial support from Capes/Faperj E-45/2013. RP thanks helpful discussions with Stefan Boettcher and Andris Ambainis' group.


\newpage

\section*{Appendix}

In this appendix we review the main graph theory concepts~\cite{Har94} we are using in this work and we give some examples.

A simple or undirected graph $\Gamma(V,E)$ is defined by a set $V$ of vertices or nodes and a set $E$ of edges so that each edge links two vertices and two vertices are linked by at most one edge. Two vertices linked by an edge are called adjacent. The complete graph is a simple graph in which every pair of distinct vertices is connected by an edge.  A directed graph is a graph whose edges have a direction associated with them.

A subgraph $\Gamma'(V',E')$, where $V'\subset V$ and $E'\subset E$, is an induced subgraph of $\Gamma(V,E)$ if it has exactly the edges that appear in $\Gamma$ over the same vertex set. If two vertices are adjacent in $\Gamma$ they are also adjacent in the induced subgraph. A clique is a subset of vertices of a simple graph such that its induced subgraph is complete. A maximal clique is a clique that cannot be extended by including one more adjacent vertex. A maximum clique is a clique of maximum possible size.  A clique can have one vertex.

A bipartite graph is a simple graph whose vertex set $V$ is the union of two disjoint sets $X$ and $X'$ so that no two vertices in $X$ are adjacent and no two vertices in $X'$ are adjacent. 

A line graph of a simple graph $\Gamma$ (called root graph) is another graph $L(\Gamma)$ so that each vertex of $L(\Gamma)$ represents an edge of $\Gamma$ and two vertices of $L(\Gamma)$ are adjacent if and only if their corresponding edges share a common vertex in $\Gamma$. For example, consider the bipartite graph $\Gamma$ of Fig.~\ref{fig:app1}. The sets of vertices $X$ and $X'$ have labels $\{\alpha_0,\alpha_1,\alpha_2,\alpha_3\}$ and $\{\beta_0,\beta_1,\beta_2\}$ respectively. The edges of $\Gamma$ have labels from 1 to 6. Each edge becomes a vertex in the line graph $L(\Gamma)$ as can be seen in Fig.~\ref{fig:app1}. Now we have to check which vertices of $L(\Gamma)$ are adjacent. For example, edges 1, 2, and 3 share a common vertex $\alpha_0$ of $\Gamma$. They are adjacent in $L(\Gamma)$ and form a clique with three vertices. This clique has label $\alpha_0$ in $L(\Gamma)$. This goes on until we have the partition described in Fig.~\ref{fig:app1}. 

\

\begin{figure}[h!] 
\centering
\includegraphics[scale=0.42]{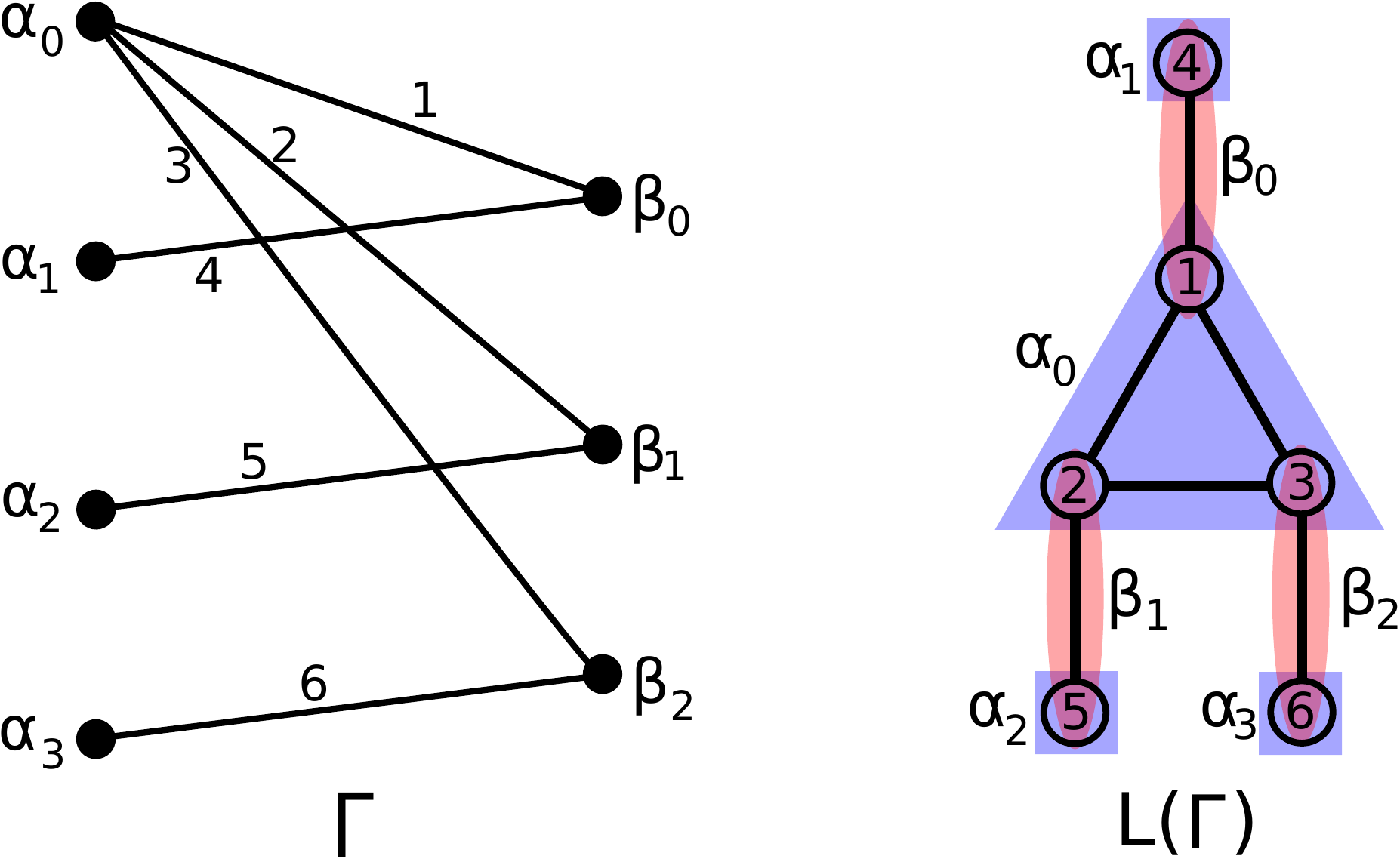}
\caption{A bipartite graph $\Gamma$ and its line graph $L(\Gamma)$ with a two-colorable clique partition labeled by the vertices of $\Gamma$.}
\label{fig:app1}
\end{figure}

Szegedy's model in $\Gamma$ uses the following vectors
\begin{align*}
	&\ket{\phi_0}=\frac{1}{\sqrt 3}\ket{\alpha_0}\left(\ket{\beta_0}+\ket{\beta_1}+\ket{\beta_2} \right),\\
	&\ket{\phi_1}=\ket{\alpha_1}\ket{\beta_0},\,\,	
	\ket{\phi_2}=\ket{\alpha_2}\ket{\beta_1},\,\,	
	\ket{\phi_3}=\ket{\alpha_3}\ket{\beta_2}
\end{align*}
and \begin{align*}
	&\ket{\psi_0}=\frac{1}{\sqrt 2}\left(\ket{\alpha_0}+\ket{\alpha_1} \right)\ket{\beta_0},\\
	&\ket{\psi_1}=\frac{1}{\sqrt 2}\left(\ket{\alpha_0}+\ket{\alpha_2} \right)\ket{\beta_1},\\
	&\ket{\psi_2}=\frac{1}{\sqrt 2}\left(\ket{\alpha_0}+\ket{\alpha_3} \right)\ket{\beta_2}.
\end{align*}
The staggered model in $L(\Gamma)$ uses the following vectors
\begin{align*}
	&\ket{\alpha_0}=\frac{1}{\sqrt 3}\left(\ket{1}+\ket{2}+\ket{3} \right),\\
	&\ket{\alpha_1}=\ket{4},\,\,	\ket{\alpha_2}=\ket{5},\,\,	\ket{\alpha_3}=\ket{6}
\end{align*}
and \begin{align*}
	&\ket{\beta_0}=\frac{1}{\sqrt 2}\left(\ket{1}+\ket{4} \right),\\
	&\ket{\beta_1}=\frac{1}{\sqrt 2}\left(\ket{2}+\ket{5} \right),\\
	&\ket{\beta_2}=\frac{1}{\sqrt 2}\left(\ket{3}+\ket{6} \right).
\end{align*}
Notice that there is a bijection between those vectors if we use the following edge-vertex correspondence: 
$\alpha_0\beta_0 \leftrightarrow 1$, 
$\alpha_0\beta_1 \leftrightarrow 2$, 
$\alpha_0\beta_2 \leftrightarrow 3$, 
$\alpha_1\beta_0 \leftrightarrow 4$, 
$\alpha_2\beta_1 \leftrightarrow 5$, 
$\alpha_3\beta_2 \leftrightarrow 6$.

If it is given two-colorable tessellations of a line graph with one vertex in each polygon intersection, it is possible to find the bipartite root graph. For example, take polygon $\alpha_0$ in $L(\Gamma)$. Selecting the polygons of tessellation $\beta$ that have overlap with $\alpha_0$, we obtain vector
$$\ket{\phi_0}=\frac{1}{\sqrt 3}\ket{\alpha_0}\left(\ket{\beta_0}+\ket{\beta_1}+\ket{\beta_2}\right),$$
and the same with respect to polygons $\alpha_1$, $\alpha_2$, and $\alpha_3$. Now, if we take $\beta_0$, polygons $\alpha_0$ and $\alpha_1$ have overlap with $\beta_0$, then
$$\ket{\psi_0}=\frac{1}{\sqrt 2} \left(\ket{\alpha_0}+\ket{\alpha_1}\right) \ket{\beta_0},$$
and the same with polygons $\beta_1$ and $\beta_2$. Following this method we re-obtain the vectors used by Szegedy's model in the first place and we re-rebuild the bipartite graph.

A Krausz partition splits the edges of the graph into cliques in such a way that no vertex belongs to more than two of the cliques. For example, the union of the red and blue tessellations of graph $L(\Gamma)$ in Fig.~\ref{fig:app1} is a Krausz partition. Notice that each edge is in one member of the partition and each vertex is exactly in two members of the partition. Krausz~\cite{Kra43} proved the following
\begin{theorem}
A graph $\Gamma'$ is a line graph of some graph if and only if $\Gamma'$ has a Krausz partition.
\end{theorem}

Fig.~\ref{fig:app2} shows a graph $\Gamma'$ that has a Krausz partition. Then $\Gamma'$ is a line graph of some graph $\Gamma$, which is depicted in the right side of the figure.  Notice again that each edge is in one member of the partition and each vertex is exactly in two members of the partition. If the Krausz partition is not two-colorable, then root graph is not bipartite. 
\begin{figure}[h!] 
\centering
\includegraphics[scale=0.38]{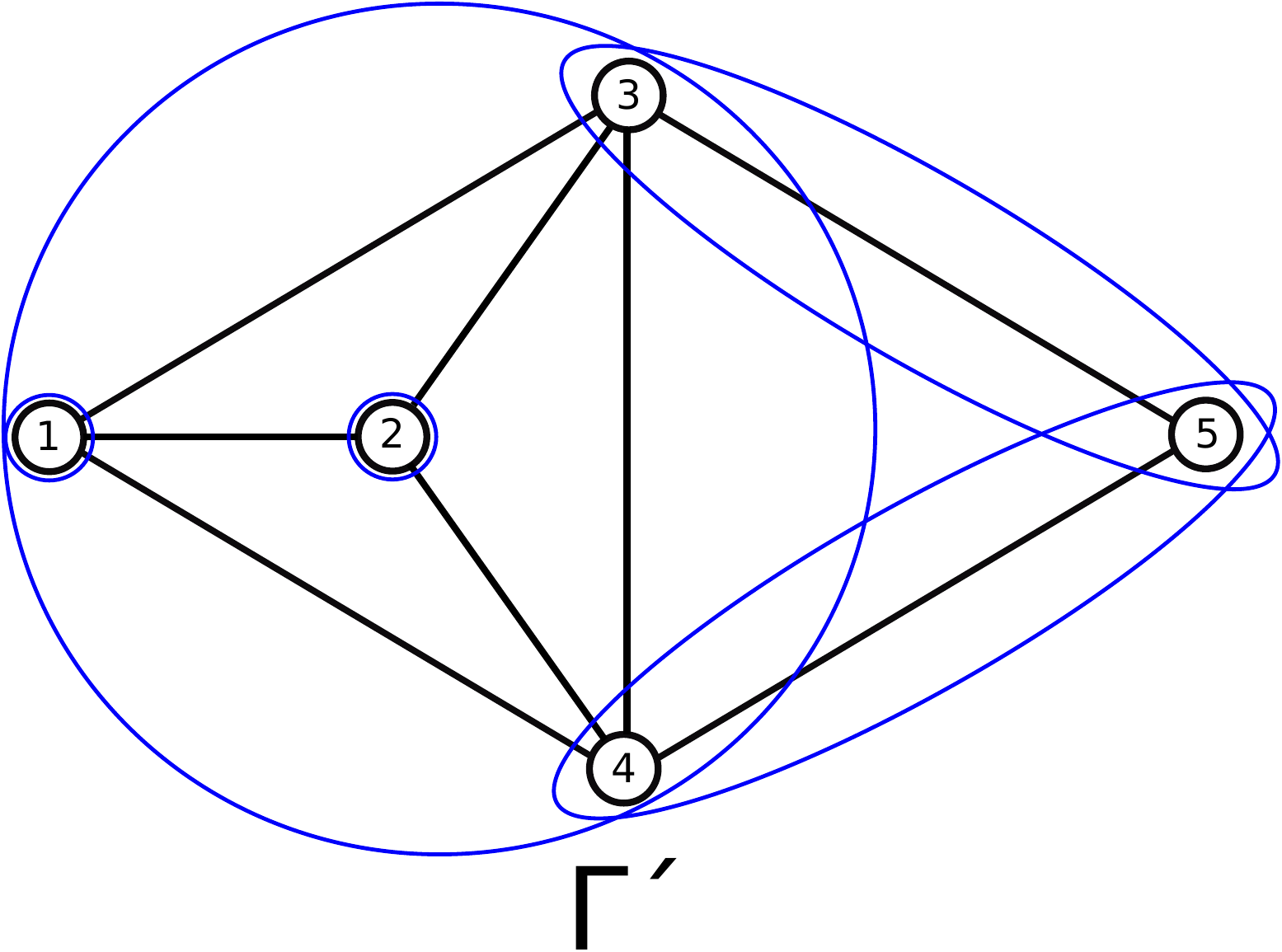}\hspace{1cm}
\includegraphics[scale=0.38]{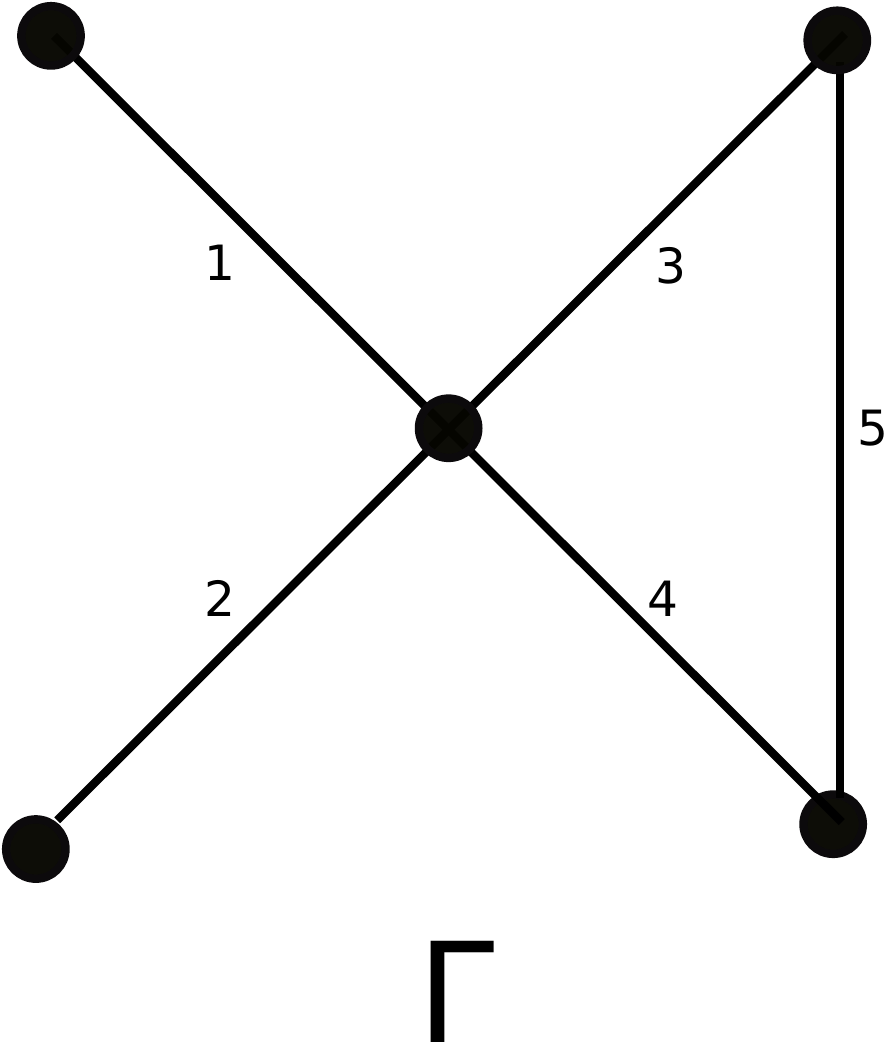}
\caption{A line graph $\Gamma'$ with the Krausz partition and its root graph $\Gamma$. Notice that the partition of  $\Gamma'$ is not two-colorable and $\Gamma$ is not bipartite.}
\label{fig:app2}
\end{figure}

\begin{figure}[b!] 
\centering
\includegraphics[scale=0.38]{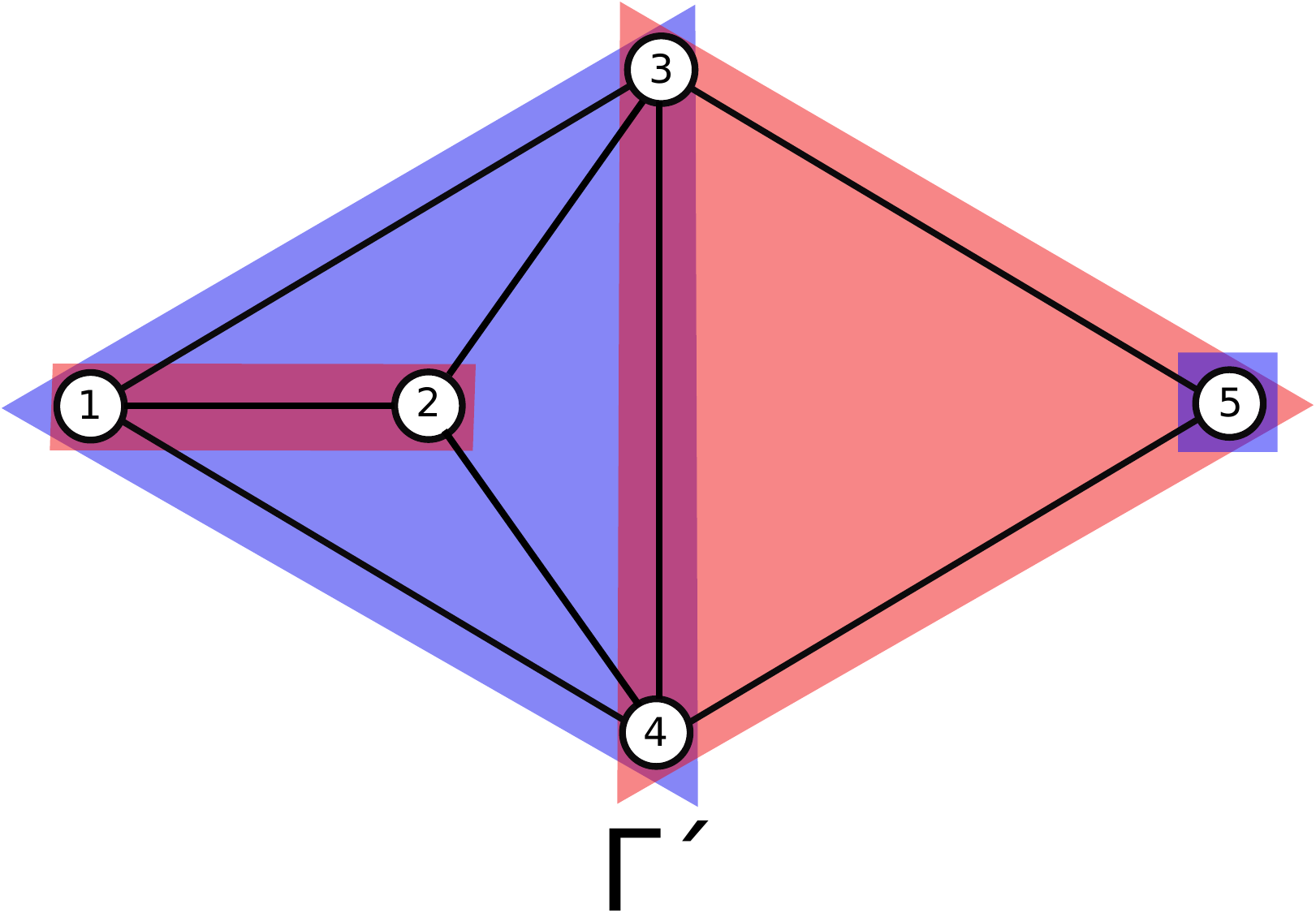} 
\caption{Tessellations of graph $\Gamma'$.}
\label{fig:app3}
\end{figure}
In order to define a staggered QW with two tessellations in graph $\Gamma'$ of Fig.~\ref{fig:app2} we must put two vertices in at least one polygon intersection. Fig.~\ref{fig:app3} shows that $\Gamma'$ is two-tessellable. In the simplest case we can take uniform vectors. The vectors of tessellation $\alpha$ (blue) are  
\begin{align*}
	&\ket{\alpha_0} = \frac{1}{2}\left( \ket{1}+\ket{2}+\ket{3}+\ket{4}\right), \\
	&\ket{\alpha_1} = \ket{5}
\end{align*}
and the vectors of tessellation $\beta$ (red) are
\begin{align*}
	&\ket{\beta_0} = \frac{1}{\sqrt 2}\left(\ket{1}+\ket{2}\right),\\
	&\ket{\beta_1} = \frac{1}{\sqrt 3}\left( \ket{3}+\ket{4}+\ket{5}\right). 
\end{align*}

\begin{figure}[h!] 
\centering
\includegraphics[scale=0.48]{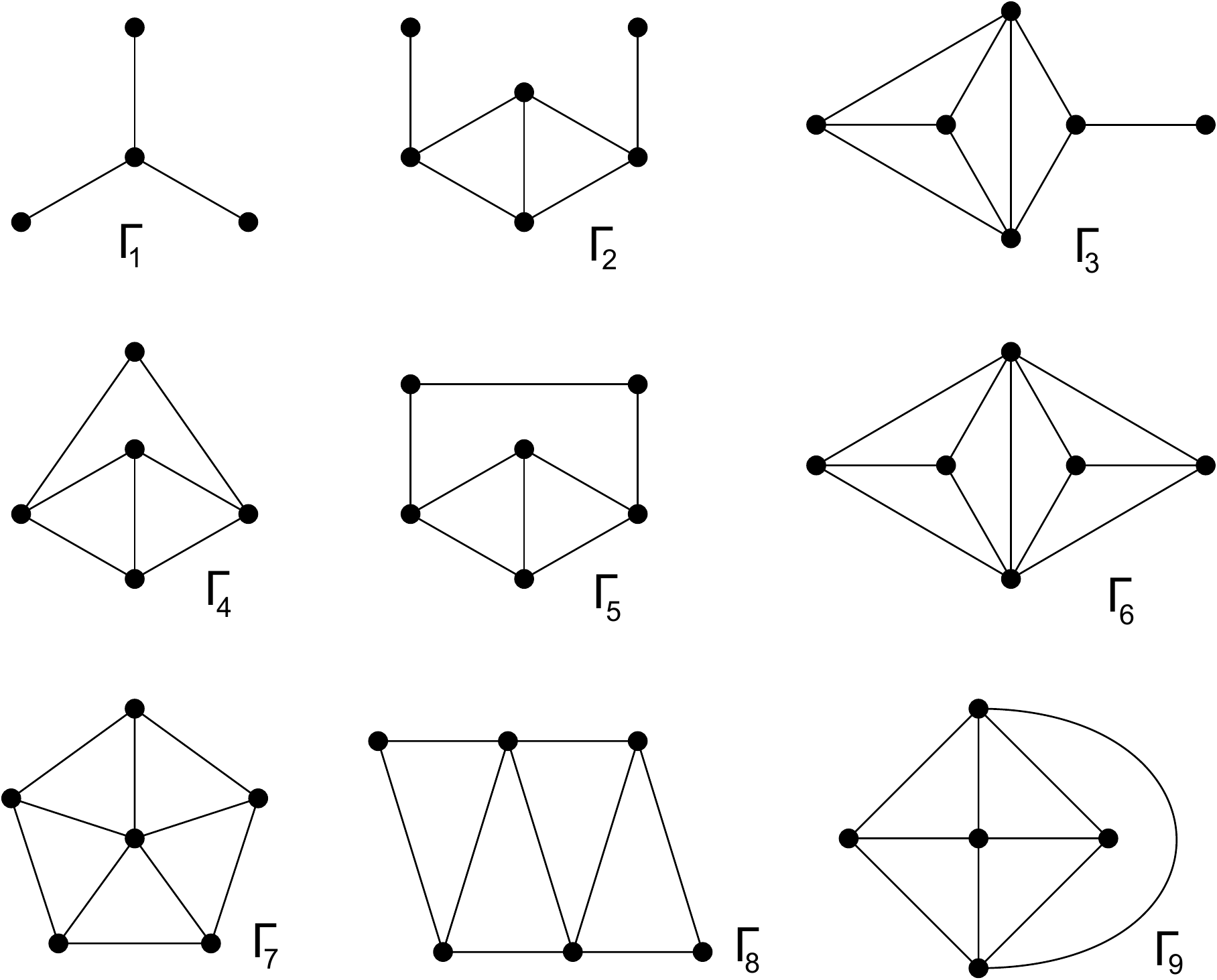}
\caption{Nine forbidden Beineke induced subgraphs.}
\label{fig:app4}
\end{figure}
It is interesting to analyze staggered QWs on graphs that are not line graphs. Those graphs do not have root graphs and any comparison with classical random walks must be performed on the same graph. Krausz partitions play no role in this case. There is an alternative form to characterize line graphs. Beineke~\cite{Bei70} proved the following
\begin{theorem}
Let $\Gamma'$ be a graph. There exists a graph $\Gamma$ such that $\Gamma'$ is the line graph of $\Gamma$ if and only if $\Gamma'$ contains no graph of Fig.~\ref{fig:app4} as an induced subgraph.
\end{theorem} 
Graph $\Gamma_1$ of Fig.~\ref{fig:app4} is called claw. A trivial consequence of Beineke's theorem is the following
\begin{corollary}
A line graph is claw-free.
\end{corollary}

By inspection, we verify that graphs $\Gamma_2$, $\Gamma_3$, $\Gamma_4$, and $\Gamma_6$ can be tessellated with two tessellations if we put two vertices in common in a polygon intersection while $\Gamma_1$, $\Gamma_5$, $\Gamma_7$, $\Gamma_8$, and $\Gamma_9$ cannot be tessellated with two tessellations. $\Gamma_5$ cannot be tessellated with two tessellations because of parity violation and the remaining graphs have three or more maximal cliques with a common vertex. To define a staggered QW on a graph that is not a line graph of a bipartite graph, either we have to use a tessellation that has at least two vertices in common in a polygon intersection or we have to use more than two tessellations.

\begin{figure}[t!] 
\centering
\includegraphics[scale=0.43]{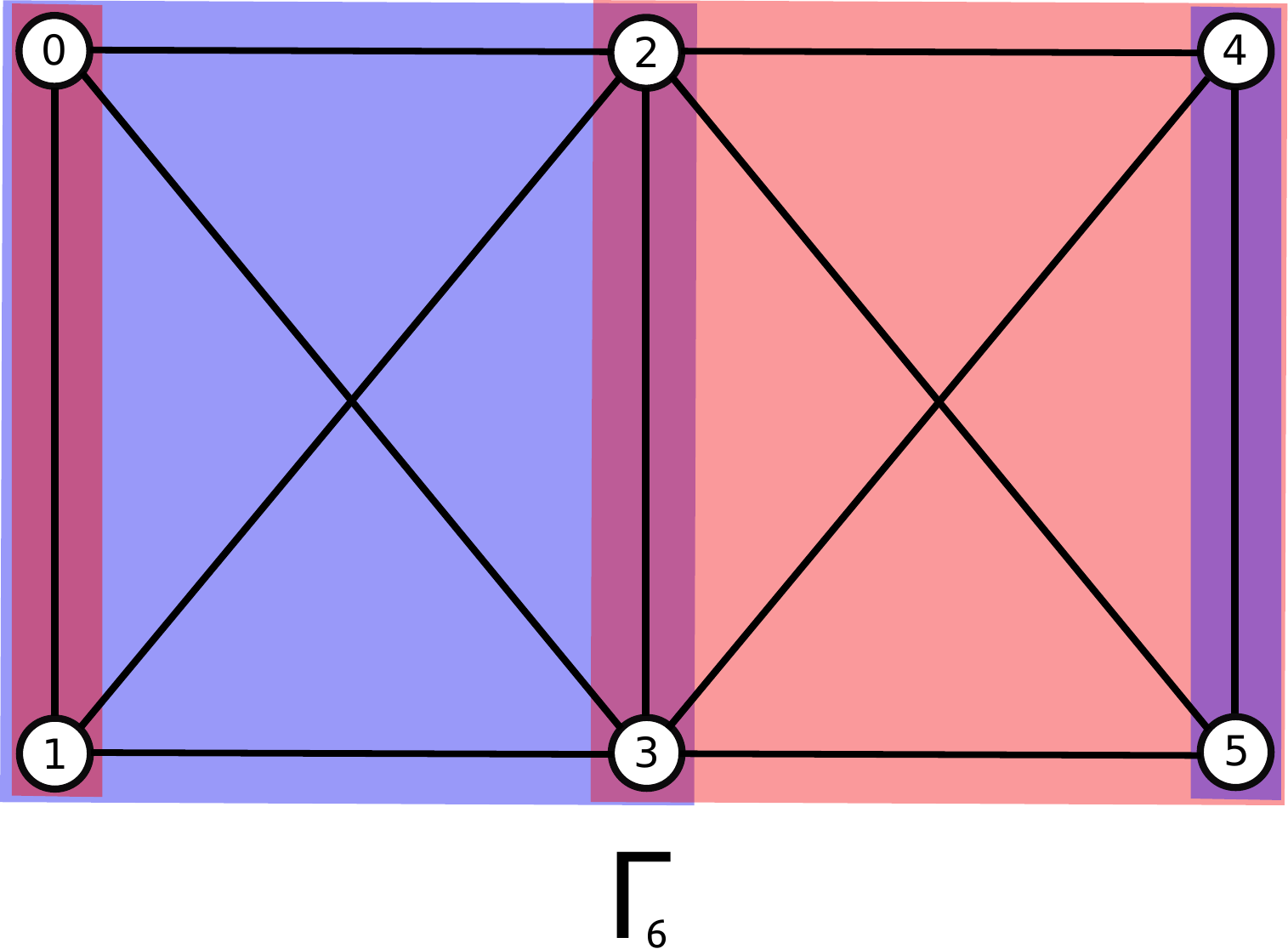}
\caption{A two-tessellable graph that is not a line graph. The tessellations have two vertices in common in at least one polygon intersection.}
\label{fig:app5}
\end{figure}
As a last example, consider the two-tessellable graph of Fig.~\ref{fig:app5}, which is not a line graph as one can check by finding one of the Beineke induced subgraphs. The vectors of tessellation $\alpha$ (blue) are  
\begin{align*}
	&\ket{\alpha_0} = \frac{1}{2}\left( \ket{0}+\ket{1}+\ket{2}+\ket{3}\right), \\
	&\ket{\alpha_1} = \frac{1}{\sqrt 2}\left(\ket{4}+\ket{5}\right)
\end{align*}
and the vectors of tessellation $\beta$ (red) are
\begin{align*}
	&\ket{\beta_0} = \frac{1}{\sqrt 2}\left(\ket{0}+\ket{1}\right),\\
	&\ket{\beta_1} = \frac{1}{2}\left( \ket{2}+\ket{3}+\ket{4}+\ket{5}\right). 
\end{align*}
This choice generates the unitary operator
$$U=\frac{1}{2}\left[ \begin {array}{cccccc} 1&-1&1&1&0&0\\ \noalign{\medskip}-1&1&1
&1&0&0\\ \noalign{\medskip}0&0&1&-1&1&1\\ \noalign{\medskip}0&0&-1&1&1
&1\\ \noalign{\medskip}1&1&0&0&1&-1\\ \noalign{\medskip}1&1&0&0&-1&1
\end {array} \right]$$
which is obtained from Eqs.~(\ref{U}) to~(\ref{beta_k}). $U$ satisfies $U^6=I$. Therefore the evolution is periodic. To obtain a more general evolution operator we can choose $\ket{\beta_0}$ as a non-trivial point in the Bloch sphere by taking
$$\ket{\beta_0} = \cos\frac{\theta}{2}\,\ket{0}+e^{-i\varphi}\sin\frac{\theta}{2}\,\ket{1},$$
where $0\le \theta\le \pi$ and $0\le \varphi < 2\pi$. 

Two-tessellable staggered QWs are interesting because we can use Szegedy's spectral theorem to find the spectrum of the evolution operator. All line graphs of bipartite graphs are two-tessellable and generate staggered QWs equivalent to Szegedy's QWs. Not all line graphs of non-bipartite graphs are two-tessellable and the same is true for graphs that are not line graphs of any graph. It would be interesting to characterize the class of two-tessellable graphs.

\end{document}